\documentclass[10pt]{article}
\usepackage[utf8]{inputenc}
\usepackage{comment}
\usepackage{graphicx}
\usepackage{amsmath, amsthm, amsfonts, amssymb}
\usepackage{verbatim}
\usepackage[dvipsnames]{xcolor}
\usepackage{wrapfig}
\usepackage{xfrac}
\usepackage{ mathrsfs }
\usepackage{multirow}
\usepackage{microtype} 
\usepackage{thmtools, thm-restate}
\usepackage{hyperref}
\usepackage{fancyhdr}
\fancypagestyle{specialfooter}{%
  \fancyhf{}
  
  \fancyfoot[L]{\footnotesize{Preprint submitted to Computational Geometry}}
}

\newcommand{\ZZ}{\mathbb{Z}}

\newcommand{\NN}{\mathbb{N}\,}

\newcommand{\PP}{\mathbb{P}}
\newcommand{\EE}{\mathbb{E}}

\newcommand{\gd}{\mathscr{D}}
\newcommand{\gdd}{\mathscr{I}}
\newcommand{\RP}{\mathcal{R}}
\newcommand{\pt}[1]{\mathsf{#1}}
\newcommand{\Mod}[1]{\ \mathrm{mod}\ #1}

 \newtheorem{theorem}{Theorem} 
 \newtheorem{lemma}{Lemma} 
 \newtheorem{corollary}{Corollary} 
  
 \newtheorem{definition}{Definition} 
  \newtheorem{proposition}{Proposition}

\title{Geometric Dominating Sets}

\author{Oswin Aichholzer$^\bot$, David Eppstein$^\circ$, Eva-Maria Hainzl\thanks{supported by the
Austrian  Science Foundation FWF, project F 55-02.}}

\date{ \small{$^\bot$ Graz University of Technology, $^\circ$ University of California, Irvine, $^\ast$ TU Wien}}

\begin{document}

\maketitle

\begin{abstract}
  We consider a minimizing variant of the well-known \emph{No-Three-In-Line Problem}, the \emph{Geometric Dominating Set Problem}:
  What is the smallest number of points in an $n\times n$~grid such that every grid point lies on a common line with two of the points in the set?
  We show a lower bound of $\Omega(n^{2/3})$ points and provide a constructive upper bound of size $2 \lceil n/2 \rceil$.
  If the points of the dominating sets are required to be in general position we provide optimal solutions for grids of size up to $12 \times 12$.
  For arbitrary $n$ the currently best upper bound for points in general position remains the obvious $2n$. Finally, we discuss the problem on the discrete torus where we prove an upper bound of $O((n \log n)^{1/2})$. For $n$ even or a multiple of 3, we can even show a constant upper bound of 4. We also mention a number of open questions and some further variations of the problem.  \\
\end{abstract}
\thispagestyle{specialfooter}

\section{Introduction}
The well-known \emph{No-Three-In-Line Problem} asks for the largest point set in an $n\times n$~grid without three points in a line. 
This problem has intrigued many mathematicians including for example Paul Erdős for roughly 100 years now. 
Few results are known and explicit solutions matching the trivial upper bound of $2n$ only exist for $n$ up to $46$ and $n = 48, 50, 52$ (see e.g.~\cite{Flammenkamp97}). 
Providing general bounds seems to be notoriously hard to solve; see~\cite{Gardner76,Hainzl20} for some history of this problem.

In this note we concentrate on two interesting minimizing variants of the No-Three-In-Line problem, which we call the \emph{Geometric Dominating Set Problem}: What is the smallest number of points, or of points in general position, in an $n\times n$~grid such that every grid point lies on a common line with two of the points in the set? The general-position variant also answers the question: What is the smallest possible output of a greedy algorithm for the No-Three-In-Line problem, that considers the points in an adversarially-chosen ordering and constructs a solution by adding points when they do not belong to lines formed by previously added pairs of points? These problems came to our attention during the 2018 Bellairs Winter Workshop on Computational Geometry. Later we found out that they were already considered in 1974 by Adena, Holton and Kelly \cite{Adena74}, and in 1976 in Martin Gardner's ``Mathematical Games'' column~\cite{Gardner76}. Gardner wrote: ``{\it Instead of asking for the maximum number of counters that can be put on an order-$n$ board, no three in line, let us ask for the minimum that can be placed such that adding one more counter on any vacant cell will produce three in line.}'' Adena et al. searched by hand for solutions to the problem for $3 \leq n \leq 10$, obtaining configurations whose sizes are $4,4,6,6,8,8,12,12$. Surprisingly, up to $n=8$, their solutions are indeed optimal, as we will see in Section~\ref{UB}.  However, it seems that no progress has been made since then, except for the special case where lines are restricted to vertical, horizontal and $45^\circ$ diagonal lines~\cite{Cooper-et-al14}.

This minimum version might remind one less of the No-Three-In-Line Problem, which itself is based on a mathematical chess puzzle, and more of the \emph{Queens Domination Problem} that asks for a placement of five queens on a chessboard such that every square of the board is attacked by a queen. In a more general setting this problem asks for the domination number of an $n\times n$~queen's graph, a graph whose vertices are chessboard squares and whose edges represent possible queen moves~\cite{Cockayne90,Ostergard-et-al01,Watkins04}. Inspired by that, we call the smallest size of a solution for the Geometric Dominating Set Problem the \emph{geometric domination number~$\gd_n$}.

\subsection{Dominating Sets}

In the spirit of mathematical chess puzzles, the Geometric Dominating Set Problem can be formulated in two variants as

\emph{How many pawns do we have to place on a chessboard such that every square lies on a straight line defined by two pawns? How many pawns do we need if no three pawns are allowed to lie on a common line?}

We will see in Section~\ref{UB} that the answer for a chessboard is eight (for both questions), and some solutions are shown in Figure~\ref{fig:pawns}. In fact, there are $228$ possibilities to do so, and 44 if we cancel out rotation and reflection symmetries.

\begin{figure}[t]
    \centering
    \begin{minipage}{0.325\textwidth}
        \centering
        \includegraphics[width=0.9\textwidth]{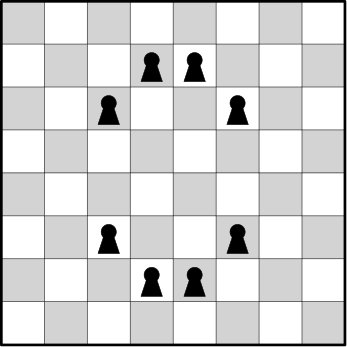}
    \end{minipage}
    \begin{minipage}{0.325\textwidth}
        \centering
        \includegraphics[width=0.9\textwidth]{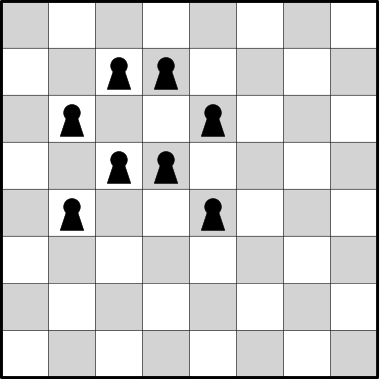}
    \end{minipage}
    \begin{minipage}{0.325\textwidth}
        \centering
        \includegraphics[width=0.9\textwidth]{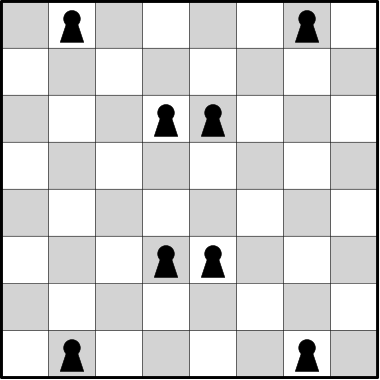}
    \end{minipage}
    \caption{Three out of 228 solutions for the $8 \times 8$ (chess) board: Every square lies on a line defined by two pawns where no three pawns are allowed to lie on a common line.}
    \label{fig:pawns}
\end{figure}
\newpage
\begin{definition}
~\\[-3ex] 
\begin{itemize}
\item Three points are called \emph{collinear} if they lie on a common straight line. Conversely, a set $S$ is called \emph{in general position} if no three points in $S$ are collinear.
\item We call a point $\pt{p}$ in the $n\times n$ grid \emph{dominated (by a set $S$)}, if $\pt{p}\in S$ or there exist $\pt{x}, \pt{y} \in S$ such that $\{\pt{x}, \pt{y}, \pt{p}\}$ are collinear. A line in the integer grid is the intersection of a straight line with the grid and we say $\pt{p}$ is dominated by a line $L$ if $\pt{p}$ is contained in~$L$. 

\item A subset $S$ of the $n\times n$ grid is called a \emph{(geometric) dominating set} or simply \emph{dominating} if every point in the grid is dominated by~$S$. 

\item We call the smallest size of a dominating set of the $n\times n$ grid the \emph{(geometric) domination number} and denote it by $\gd_n$.

\item The smallest size of a dominating set in general position (an \emph{independent dominating set}) is called the \emph{independent (geometric) domination number} and denoted by $\gdd_n$.
\end{itemize}
\end{definition}

Note that every point in an independent dominating set is only dominated by pairs that include the point itself.

\subsection{Summary Of Results}

We will show that
\begin{itemize}
    \item $\gdd_n \geq \gd_n = \Omega(n^{2/3})$ (Subsection~\ref{GDLow}, Theorem~\ref{thm:lowbound}), and
    \item $\gd_n \leq 2 \lceil n/2 \rceil$ (Subsection~\ref{UB}, Theorem~\ref{thm:upGD}).
\end{itemize}
Additionally, we will present several computational results on dominating sets of the $n\times n$ grid.
In Section \ref{sec:torus} we consider the problem on the discrete torus, where we denote the domination number by $\gd_n^T$. We prove:
\begin{itemize}
    \item For $n$ prime, $\gd_n^T = O(\sqrt{n \log n})$ and $\gd_n^T = \Omega (\sqrt{n})$.
    \item For $n = pq$, where $p$ prime and $q
    \geq 2$, we show $\gd_n^T \leq 2(p+1)$. Further, for $n$ even or divisible by 3, it holds that $\gd_{n}^T \leq 4$.
\end{itemize}

\section{Lower Bounds on $\gd_n$ and $\gdd_n$}\label{GDLow}
For a lower bound on $\gd_n$, let us consider a set $S$ of $s$ points in the $n\times n$ grid. Any pair of points in $S$ can dominate at most $n-2$ other points, so it has to hold that
${s \choose 2}(n-2)+s \geq n^2$ which is untrue when $s\le \sqrt{2n}$. Therefore, $\gd_n = \Omega(n^{1/2})$. (See \cite[Lemma 9.15]{Eppstein18}.)

However, hardly any lines in the $n\times n$ grid dominate $n$ points. In fact, we can prove a significantly better bound by using the following theorem, where $\varphi(i)$ denotes the Euler totient function, the number of positive integers less than~$i$ that are relatively prime to~$i$.

\begin{restatable}{theorem}{NumberOfDominatedPoints}\label{domptsUP}
Let $n=2k+1$ and $S$ be a subset of the $n\times n$ grid with $|S| \leq 4 \sum_{i=1}^m \varphi(i)$, where $1\leq m \leq k$. Then the number of points dominated by lines incident to a fixed point $\pt{p} \in S$ and the other points in $S$ is bounded by
\[1+8 \sum_{i=1}^{m} \left\lfloor \frac{n}{i} \right\rfloor \varphi(i) \;\leq\;\frac{48}{\pi^2}nm + O\left(n\log m\right).\]
\end{restatable}

The proof of this theorem requires the following well known number theoretic result.
\begin{restatable}[Arnold Walfisz~\cite{Walfisz63}]{theorem}{ThmWalfisz}\label{walf1}
    \[\sum_{i=1}^k \varphi(i) = \frac{3}{\pi^2}k^2 + O\left(k(\log k)^\frac{2}{3}(\log \log k)^\frac{4}{3}\right)\]
    \[\sum_{i=1}^m \frac{\varphi(i)}{i} = \frac{6}{\pi^2}m + O\left((\log m)^\frac{2}{3}(\log \log m)^\frac{4}{3}\right)\]
\end{restatable}

\begin{proof}[Theorem \ref{domptsUP}]
First we compute the maximum number of dominated points by $4t$  lines incident to the point~$\pt{c}_n = (k+1, k+1)$ in the center of the $n\times n$ grid. Since the grid can be seen as the union of two rotated copies plus two mirrored and rotated copies of the grey area in the first picture of Figure~\ref{fig:lines2}, we only need to consider all lines through~$\pt{c}_n$ that are also incident to a point in the area 
\[A_n = \{(x,y) \in \{1,2,\dots,n\} \times \{1,2,\dots,n\}  \mid k+1 \leq y \leq x\}.\] 

The general idea is to choose the $t$ lines going through a point in~$A_{n}$ that dominate the most points and the main challenge is to distinguish the lines depending on the number of points that they dominate. 
With this in mind, we observe that we can identify lines with the closest dominated point $(x,y)$ to $\pt{c}_n$ (see Figure~\ref{fig:lines2} where this point is marked red). Since $(x,y)$ is closest to $\pt{c}_{n}$ it has to hold that $\gcd\left(x-(k+1),y-(k+1)\right)=1$. Therefore, $(y-(k+1))/(x-(k+1))$ is the slope of the line given as a reduced fraction.
Moreover, observe that a line that we identify with $(x,y)$, where $x-(k+1)=j$, will dominate $2 \left\lfloor  \frac{k}{j} \right\rfloor + 1$ points. Hence, there are as many lines as numbers smaller than $j$ and coprime to $j$ which dominate $2 \left\lfloor  \frac{k}{j} \right\rfloor  + 1$ points. That is, there are exactly $\varphi(j)$ lines, where $\varphi(j)$ denotes the Euler-Phi function, each dominating $2 \left\lfloor  \frac{k}{j} \right\rfloor  + 1$ points.

\begin{figure}[th]
    \centering
    \begin{minipage}{0.32\textwidth}
        \centering
        \includegraphics[width=0.9\textwidth]{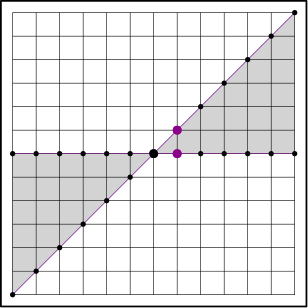}\\ \vspace{4mm}
        \includegraphics[width=0.9\textwidth]{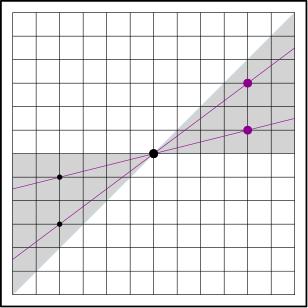} 
    \end{minipage}
    \begin{minipage}{0.32\textwidth}
        \centering
        \includegraphics[width=0.9\textwidth]{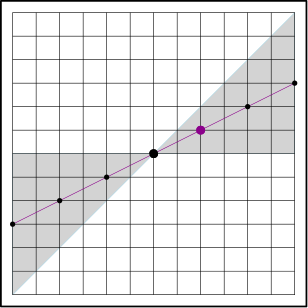} \\ \vspace{4mm}
        \includegraphics[width=0.9\textwidth]{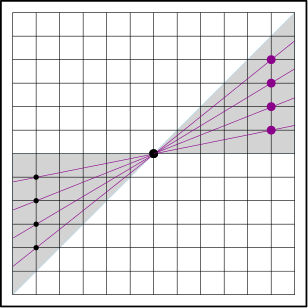} 
    \end{minipage}
    \begin{minipage}{0.32\textwidth}
        \centering
        \includegraphics[width=0.9\textwidth]{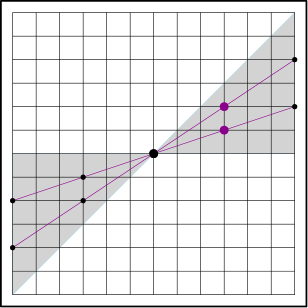} \\ \vspace{4mm}
        \includegraphics[width=0.9\textwidth]{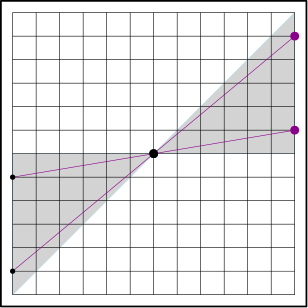} 
    \end{minipage}
    \caption{Example for $n=13$: Counting the lines incident to the center point $\pt{c}_n$, by their slope}\label{fig:lines2}
\end{figure}

Summing over all four areas, we choose $4\sum_{i=1}^{k}\varphi(i)$ lines and these lines dominate at most $1+4\sum_{i=1}^{k} 2 \left\lfloor \frac{k}{j} \right\rfloor \varphi(i)$ points, counting $\pt{c}_{n}$ only once.

Applying Theorem \ref{walf1} to our computations above, we obtain that the number of points dominated by lines incident to $\pt{c}_n$ and some point in $S$ is bounded by
\[1+8 \sum_{i=1}^{m} \left\lfloor \frac{k}{i} \right\rfloor \varphi(i) \;\leq\; 1+8 \left(\frac{n-1}{2}\right)\sum_{i=1}^{m} \frac{\varphi(i)}{i}  = \frac{24}{\pi^2}nm + O\left(n\log m\right).\]

But what if we fix a point $\pt{p}$ that is not in the center of the $n\times n$ grid?
If we consider the $(2n+1)\times (2n+1)$ grid with $\pt{p}$ at the center, then we know that the original $n\times n$ grid is a subset of this grid. Hence, the upper bound on the number of dominated points in the $(2n+1)\times (2n+1)$ grid will also be an upper bound on the number of dominated points in this subset.
\end{proof}

\begin{theorem}[A lower bound on $\gd_{n}$] \label{thm:lowbound} For $n\in \mathbb{N}$, it holds that $\gd_n = \Omega(n^{2/3})$.
\end{theorem}

\begin{proof}
First, let $n=2k+1,\; k\in\mathbb{N}$ and let $S$ be a set of $s$ points in the grid, where $\sqrt{2n} \leq s\leq 2n$. (Recall that $2n$ is a trivial upper bound on $\gd_n$ and $\sqrt{2n}$ a lower bound.)
Let $m$ be such that
\[4\cdot\sum_{i=1}^{m-1} \varphi(i) < s \leq 4\cdot\sum_{i=1}^m \varphi(i)\]
Then $s = \frac{12}{\pi^2}m^2 + O(m \log m)$ by Theorem~\ref{walf1}.

By Theorem~\ref{domptsUP}, the number of points dominated by lines incident to a fixed point $\pt{p}$ and one of $s-1$ additional points is bounded by
$\frac{48}{\pi^2}nm + O\left(n\log m\right)$. To dominate all points in the grid, we thus need \[n^2 \leq s \left(\frac{48}{\pi^2}nm + O\left(n\log m\right)\right).\]

Next, we plug in the asymptotic expression for $s$, such that the inequality simplifies to
\[n^2 \leq \left(\frac{12}{\pi^2}m^2 + O\left(m\log m\right)\right) \left(\frac{48}{\pi^2}nm + O\left(n\log m\right)\right) =\frac{576}{\pi^4}nm^3 + O\left(nm^2\log m\right)\]

If we divide by $n$, we can see that $m = \Omega(n^{1/3})$ and consequently $s = \Omega(n^{2/3})$ which proves the claim for $n$ odd.

For $n$ even we embed the $n\times n$ grid into the $(n+1)\times (n+1)$ grid and obtain the same asymptotic results.
\end{proof}

\begin{corollary}
$\gdd_n = \Omega(n^{2/3})$.
\end{corollary}
\begin{proof}
Since any independent dominating set is a dominating set, we have $\gdd_n \geq \gd_n$.
\end{proof}

\section{Upper Bounds on $\gd_n$}\label{UB}

\begin{theorem}\label{thm:upGD}
For $n\geq 3$, it holds that $\gd_n \leq 2 \left\lceil \frac{n}{2} \right\rceil$.
\end{theorem}

\begin{figure}[ht]
    \centering
        \includegraphics[width=0.4\textwidth]{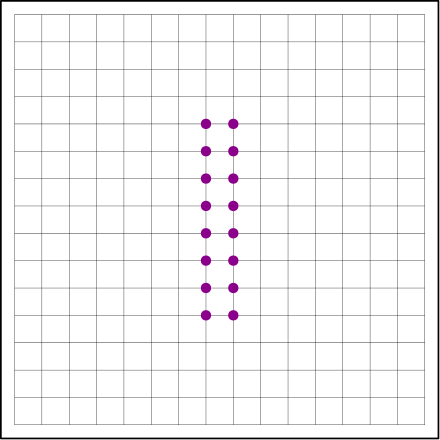}\hspace{10mm}
        \includegraphics[width=0.4\textwidth]{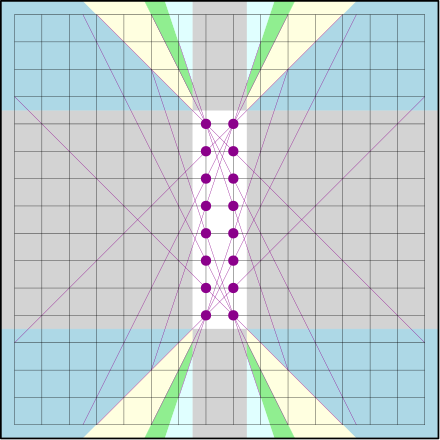}
    \caption{Dominating set construction for $n=16$.}
    \label{Fig:Upbound}
\end{figure}
\begin{proof}
First, let us consider $n = 2k$. We denote the range from $1$ to $n$ by $[n] = \{1,2,3,\dots, n\}$. The idea is to choose $k$ points on each of the two vertical lines in the middle as dominating sets (see Figure~\ref{Fig:Upbound}). That is, 
\[S = \{(i,j) \mid k \leq i \leq k+1,\; \lceil k/2 \rceil+1 \leq j \leq k + \lceil k/2 \rceil \}.\]
Obviously, all $\pt{x}=(x_1,x_2)$ with $x_1 \in \{k,k+1\}$ are dominated by a vertical line and those with $\lceil k/2 \rceil+1 \leq x_2 \leq k + \lceil k/2 \rceil$ are dominated by a horizontal line.
By symmetry it is sufficient to prove that points in the lower left rectangle $[k-1]\times[\,\lceil k/2 \rceil\,]$ are dominated.

So let $\pt{x} = (x_1,x_2) \in [k-1]\times[\,\lceil k/2 \rceil\,]$ and $\Delta = k-x_1$. We will now try to find the line with the smallest slope $t$ that is incident to $\pt{x}$ and two points in $S$. In Figure \ref{Fig:Upbound}, all points that are dominated by a line with the same smallest slope, lie in the same colored area.

Let $t$ be the smallest positive integer such that $x_2+t\, \Delta > \lceil k/2 \rceil$. Since $1 \leq \Delta \leq k-1$, this integer $t$ is well defined and in the range from $1$ to $k$. This $t$ is in fact the smallest slope of a line that dominates $\pt{x}$. What is left to show is, that the points $(k, x_2+t \,\Delta)$ and $(k+1,x_2+t\, (\Delta+1))$ are in $S$. This is the case if $x_2+t\, (\Delta+1) \leq k + \lceil k/2 \rceil$. We consider two cases:
\begin{enumerate}
    \item $\Delta = k-x_1 \geq \lceil k/2 \rceil$. Then $t=1$ and $x_2+\Delta+1 \leq k + \lceil k/2 \rceil$, since $\Delta \leq k-1$ and $x_2 \leq \lceil k/2 \rceil$. (Note that this case is tight, which is why we cannot extend the grid).
    \item $\Delta < \lceil k/2 \rceil$. Since $t$ is the smallest integer such that $t > (\lceil k/2 \rceil - x_2)/\Delta $, we know that $t \leq (\lceil k/2 \rceil - x_2)/\Delta +1$ and obtain  
    \begin{align*} x_2+t\, (\Delta+1) &\leq x_2 + \left( \frac{\lceil k/2 \rceil - x_2}{\Delta} +1 \right)(\Delta+1) \\
    &= \lceil k/2 \rceil + \frac{\lceil k/2 \rceil - x_2}{\Delta}+(\Delta+1) \\
    &\leq \lceil k/2 \rceil + \frac{\lceil k/2 \rceil - 1}{1}+(\lceil k/2 \rceil-1+1) \\
    &= 3\lceil k/2 \rceil-1 < k + \lceil k/2 \rceil + 1
    \end{align*}
\end{enumerate}
Thus, $(k,x_2+t\, \Delta)$ and $(k+1,x_2+t\, (\Delta+1))$ are indeed in $S$ and $x$ is dominated by the line that is defined by the two points.

Finally, if $n = 2k-1$, we can embed the $n \times n$ grid in the $2k \times 2k$ grid and obtain the desired upper bound.
\end{proof}

If $n=k^2$ is an odd square the result can be slightly improved to $n-1$ by a construction similar to the one depicted for $k=3$ and $n=9$ in the leftmost drawing of Figure~\ref{fig:n0910collinearexterior}.

So far, for $\gdd_n$ there is no better upper bound known than the obvious $2n$. 

\section{Small Cardinalities and Examples}

\begin{figure}[ht]
	\centering{
	\includegraphics[page=1,scale=0.4]{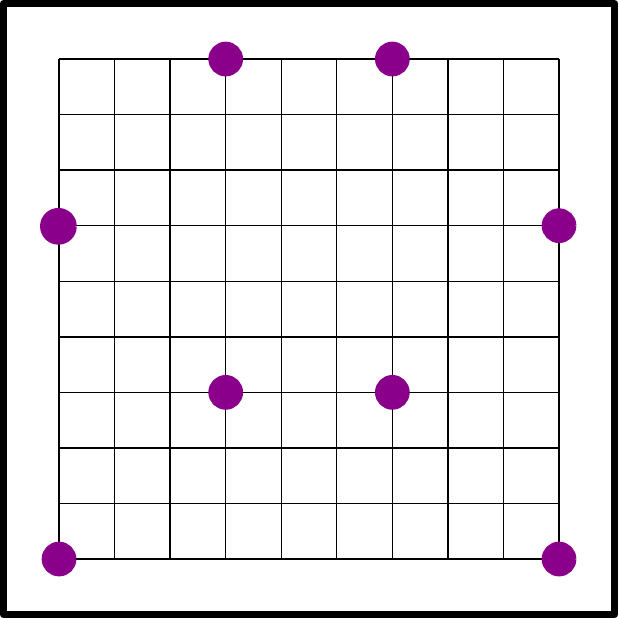}\hspace{2mm}
	\includegraphics[page=1,scale=0.4]{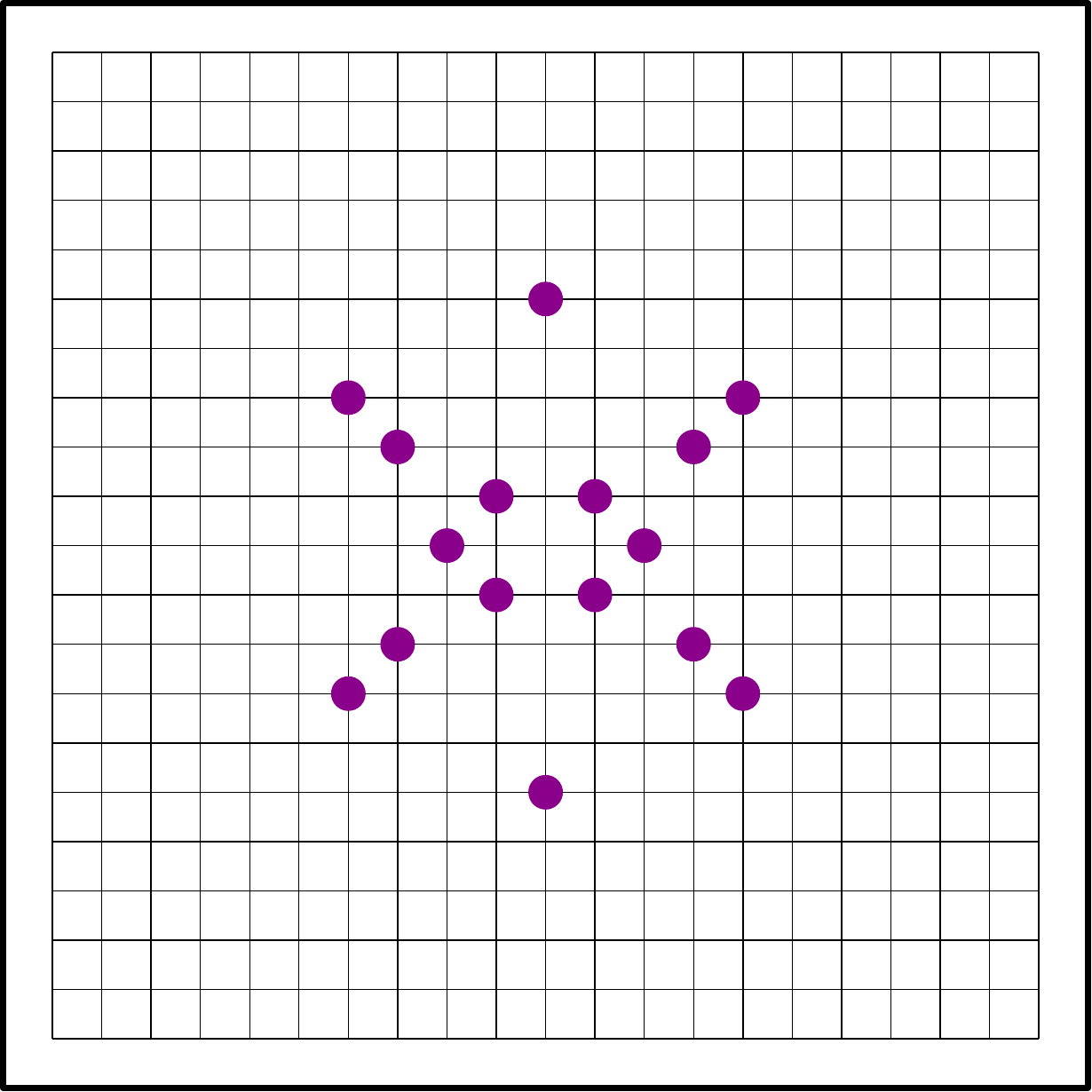}
	}
     \caption{\label{fig:examples1021} The unique (up to symmetry) minimal independent dominating set of size 8 for the $10 \times 10$ board and a small independent dominating set of size 16 for the $21 \times 21$ board. The latter gives the currently best known ratio (number of points / grid size) of $16/21<0,762$.}
\end{figure}

To obtain results for small grids we developed a simple search algorithm based on the classic backtracking approach. To speed up the computation, both symmetries -- rotation and reflection -- were taken into account. For $n=2,\ldots ,12$ we made an exhaustive enumeration of all independent dominating sets. The obtained results are summarized in Table~\ref{tab:soltable}. For larger sets upper bounds on $\gdd_n$ are given in Table~\ref{tab:largersets}.
Figure~\ref{fig:examples1021} gives two examples of small independent dominating sets.

\begin{table}[htb]
\begin{center}
 \begin{tabular}{ r | r | r | r | r | r | r | r | r | r | r | r }
 \hline
 $n$ & 2 & 3 & 4 & 5 & 6 & 7 & 8 & 9 & 10 & 11 & 12 \\
 \hline \hline
 $\gdd_n$ & 4 & 4 & 4 & 6 &  6 &  8 & 8 & 8 & 8 & 10 &  10 \\
 \hline
 diff. sets & 1 & 5 & 2 & 152 & 8 & 4136 & 228 & 11 & 4 & 108 & 12 \\
 \hline
 non sym. sets & 1 & 2 & 2 & 26 & 2 & 573 & 44 & 3 & 1 & 19 & 2 \\ 
 \hline
\end{tabular}
\caption{Size $\gdd_n$ of smallest independent dominating sets for $n=2,\ldots ,12$ and the number of different such sets. The last row shows the number of different sets if we consider symmetry by rotation and/or reflection.}
\label{tab:soltable}
\end{center}
\end{table}

\begin{table}[htb]
\setlength{\tabcolsep}{4.5pt}
\begin{center}
 \begin{tabular}{ r | r | r | r | r | r | r | r | r | r | r | r | r | r | r | r | r | r | r }
 \hline
 $n$ & 13 & 14 & 15 & 16 & 17 & 18 & 19 & 20 & 21 & 22 & 23 & 24 & 25 & 26 & 27 & 28 & 29 & 30 \\
 \hline \hline
 $\gdd_n \leq $ & 12 & 12 & 14 & 14 & 15 & 16 & 16 & 16 & 16 & 18 & 20 & 20 & 22 & 24 & 24 & 24 & 24 & 25 \\
 \hline
\end{tabular}
\caption{Currently best upper bounds for smallest independent dominating sets for $n>12$.}
\label{tab:largersets}
\end{center}
\end{table}

\begin{figure}[h!]
	\centering{
	\includegraphics[page=1,scale=0.2]{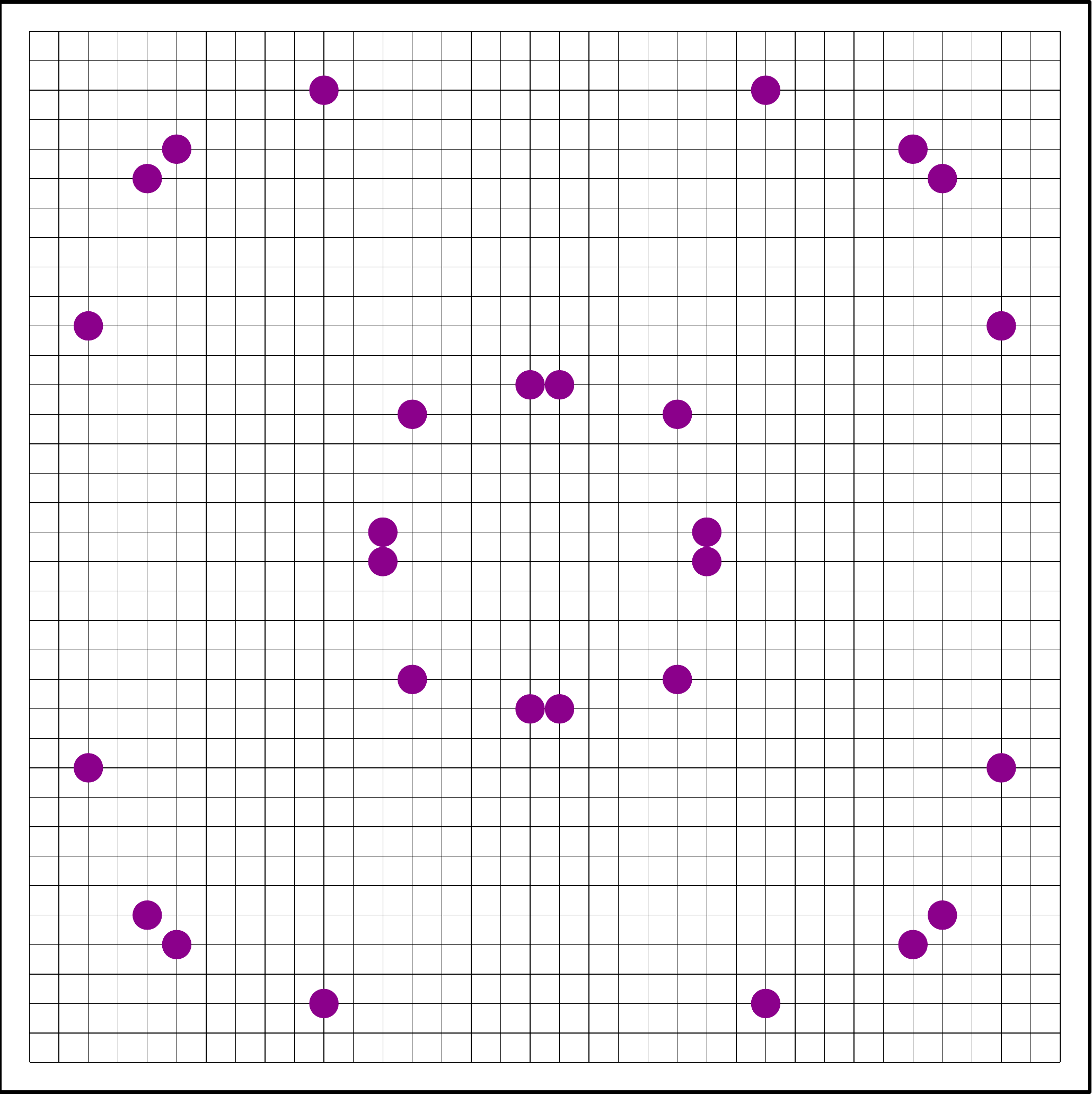}\hspace{2mm}
	\includegraphics[page=2,scale=0.2]{n36_k28.pdf}\hspace{2mm}
	\includegraphics[page=3,scale=0.2]{n36_k28.pdf}
	}
     \caption{\label{fig:large} Three independent dominating sets  of cardinality 28 for a $36 \times 36$ board.}
\end{figure}

We also obtained results for larger sets, but there is no evidence that our sets are (near the) optimal solutions. Most of these examples are rather symmetric, but that might be biased due to the approach we used to generate larger sets from smaller sets by adding symmetric groups of points. Figure~\ref{fig:large} shows three, kind of aesthetically appealing, drawings for $n=36$ with independent dominating sets of size 28.

\begin{figure}[ht]
	\centering{
	\includegraphics[page=1,scale=0.5]{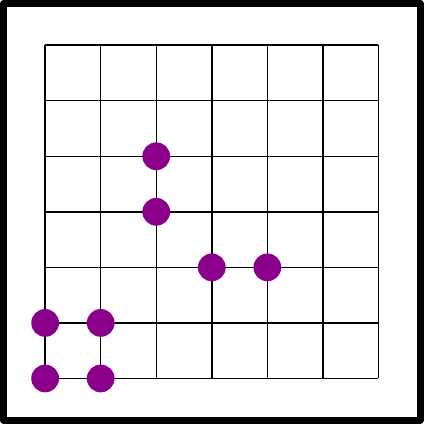}\hspace{2mm}
	\includegraphics[page=2,scale=0.5]{n07_k0708.pdf}\hspace{2mm}
	\includegraphics[page=3,scale=0.5]{n07_k0708.pdf}\hspace{2mm}
	\includegraphics[page=4,scale=0.5]{n07_k0708.pdf}\hspace{2mm}
	\includegraphics[page=5,scale=0.5]{n07_k0708.pdf}
	}
     \caption{\label{fig:collinear} Five different dominating sets for a $7 \times 7$-board. The first two sets are in general position and have size 8, while the remaining three sets have size 7 but contain collinear points.}
\end{figure}

Figure~\ref{fig:collinear} shows different dominating sets for $n=7$. The best dominating sets that contain collinear points are smaller than the best solutions in general position. For $n \leq 12$ this is the only board size where allowing collinear points leads to smaller dominating sets. Figure~\ref{fig:n0910collinearexterior}(above) shows  some nicely symmetric dominating sets with collinear points.

\begin{figure}[ht]
	\centering{
	\includegraphics[page=1,scale=0.45]{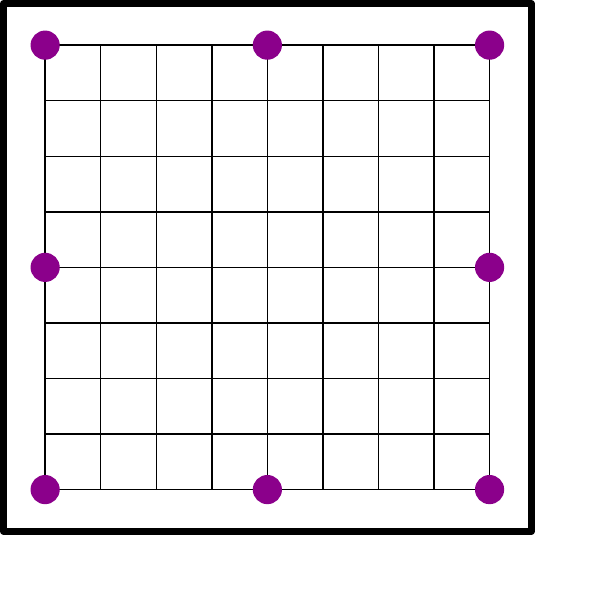}
	\includegraphics[page=2,scale=0.45]{n0910_collinear.pdf}
	\includegraphics[page=3,scale=0.45]{n0910_collinear.pdf}\hspace{12mm}
	\includegraphics[page=1,scale=0.45]{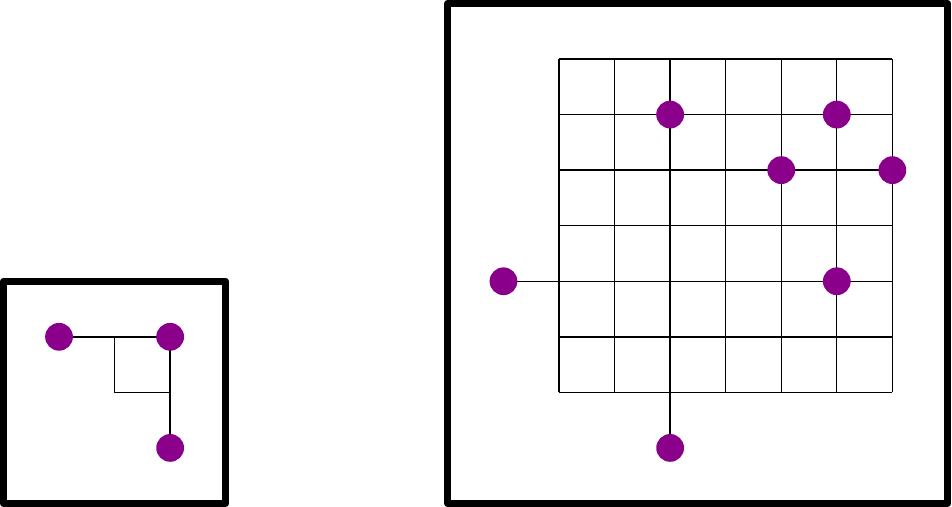}
	}
     \caption{\label{fig:n0910collinearexterior} Above: Symmetric dominating sets with collinear points for $n=9$ and $n=10$. Below: Smaller dominating sets for a $2 \times 2$-board and a $7 \times 7$-board if points are allowed to be outside of the board. These solutions are unique up to symmetry.}
\end{figure}

We can also release the restriction that the points of the dominating set have to lie within the grid, that is, the points can have coordinates smaller than one, or larger than $n$. In Figure~\ref{fig:n0910collinearexterior}(below) we depict two examples where the shown dominating sets are smaller than the best bounded solutions in general position. So far we have not been able to find any examples where this idea combined with collinear points in the dominating set provided even better solutions.

\section{Dominating sets of the Discrete Torus}\label{sec:torus}

In this section we consider the geometric dominating set problem on the discrete $n\times n$ torus. 

\begin{definition}
 We identify the discrete $n\times n$ torus $T_n$ with $\{0,1,2,\dots,n-1\}\times\{0,1,2,\dots,n-1\}$ and define lines on the torus to be images of lines in $\ZZ \times \ZZ$ under the projection $\pi_n : \ZZ \times \ZZ \rightarrow [n]^2$ defined as follows
\[\pi_n(x_1, x_2) := (x_1 \Mod n, x_2 \Mod n)\]
where $a \Mod b$ means the smallest non-negative remainder when $a$ is divided by~$b$. \\
A set of points $S$ on the torus is collinear, if it is contained in a line on the torus. A (geometric) dominating set of the torus is analogously defined as in the plane grid and we denote the domination number of $T_n$ by $\gd_n^T$.
\end{definition}

Note that the discrete torus is usually defined as the set \mbox{$\ZZ \slash n\ZZ \times \ZZ \slash n\ZZ$}, which is a group and the lines on the torus are cosets of maximal cyclic subgroups. However, we prefer the projection to the smallest non-negative remainder as representative and working with integer coordinates.\\

In Figure~\ref{fig:toruslin}, we can see that the line incident to $(0,0)$ with slope~$\frac{1}{2}$ in the $n\times n$ grid "wraps around the torus" such that it is incident to the $13$ black points. For instance, the line incident to $(1,7)$ with slope~$7$ now projects to exactly the same line.

\begin{figure}[h]
    \centering
    \includegraphics[width=0.65\textwidth]{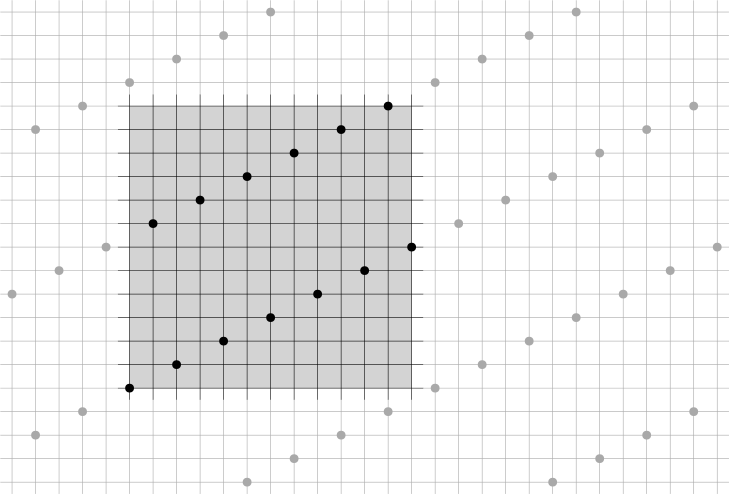}
    \caption{A sample line in black on the $13\times 13$ torus}
    \label{fig:toruslin}
\end{figure}

\subsection{Lower Bounds} \label{sec:torlow}
Roughly ten years ago, Fowler, Groot, Pandya and Snapp were the first to consider the No-Three-In-Line Problem on the $n\times n$ torus and proved for $n$ prime that the maximal size of a point set in general position is $n+1$ by giving an explicit construction and proving its maximality~\cite{Fowler-et-al12}.
In particular, they showed in the proof of their Theorem 2.3 that every point on the $n\times n$ torus is incident to exactly $n+1$ distinct lines if $n$ is prime and they explicitly identified the \emph{generators} $G = \{(0,1), (1,0), (1,1), (1,2),\dots (1,n-1)\}$ of the lines incident to $(0,0)$. While they used the term generators in the algebraic sense, we think about these generators in a similar way as we thought about the closest point to $\pt{c}_n$ on a line in $[n]^2$ in the proof of Theorem \ref{domptsUP}. These generators uniquely define a line incident to $(0,0)$ and every point on the line generated by $(x_1,x_2)$ is of the form 
\[\pi_{n}(tx_1,tx_2) \text{, where } 1\leq t \leq n.\]

Since $n$ is prime, we know that for $1\leq t_1 < t_2 \leq n$ and $x \in [n-1]$ it holds that
\[t_1x \text{ mod } n \neq t_2 x \text{ mod } n.\] 
It follows for any $(x_1,x_2) \in G$ that 
\[\pi_{n}(t_1x_1,t_1x_2) \neq \pi_{n}(t_2x_1,t_2x_2).\]
Thus, every line is incident to exactly $n$ points.

Even though it seems obvious, the property that the intersection of two lines contains at most one point is due to the fact that $n$ is prime and can be easily verified by counting the points on the torus line by line incident to $(0,0)$. 
However, if we consider the lines generated by $(1,1)$ and $(1,3)$ on the $4\times 4$ torus, we can see that they intersect not only in $(0,0)$ but also in $(2,2)$.
In fact, we will later prove that for the $2^n \times 2^n$ torus, the domination number is at most $4$.\\

This is why we repeat the proof of \cite[Lemma 9.15]{Eppstein18} here to make clear that the lower bound on $\gd_{n}^{T}$ still holds for the torus if $n$ is prime.

\begin{theorem}\label{thm:tordown}
    For $n$ prime, the domination number of $T_n$ is $\gd_n^T = \Omega(\sqrt{n})$
\end{theorem}
\begin{proof}
    If $n$ is prime, any two distinct points of $T_n$ define exactly one line that dominates $n-2$ other points. If $S$ is a dominating set of size $s$, it therefore has to hold that 
    \[\binom{s}{2}(n-2)+s \geq n^2\]
    which implies that $s > \sqrt{2n}$.
\end{proof}

\subsection{Upper Bounds}
By extending the probabilistic approach of Guy and Kelly to the No-3-In-Line problem~\cite{Guy-et-al68} in the plane, we show an upper bound of $O(\sqrt{n\log n})$, which remarkably is below the lower bound of the regular grid. If $n$ is not prime, we can do even better.

In \cite{Guy-et-al68}, Guy and Kelly calculated the probability that three uniformly at random chosen points in $[n]^2$ lie on a line and used the fact that under the assumption that the events that three points of $\alpha n$ random points of $[n]^2$ are in general position are mutually independent, the probability that $\alpha n$ random points are in general position is just the product over the probabilities of the single events. That is, the probability would be
\begin{equation*}
\left(1-\frac{18\log n}{\pi^2n^2} +
O\left(\frac{1}{n^2}\right)\right)^{\binom{\alpha n}{3}} = \exp\left(-\frac{3\alpha ^3}{\pi^2}n\log n + O(n)\right).
\end{equation*}
If this would be the correct probability that a random point set of size $\alpha n$ is in general position, then the number of solutions to the No-Three-In-Line Problem would have to be
\begin{equation*}
\binom{n^ 2}{\alpha n}\exp\left(-\frac{3\alpha ^3}{\pi^2}n\log n + O(n)\right) = O\left(n^{\alpha -\frac{3\alpha ^3}{\pi^2}n}c^n\right)
\end{equation*}
where $c$ is a large enough constant and which converges to $0$ as $n\rightarrow \infty$ if $\alpha  > \left(\pi^2/3\right)^{1/2} \approx 1.813799$.\footnote{In fact, they obtained a slightly different result due to a computational error. Gabor Ellman pointed out the mistake to Richard Guy in a personal conversation and Ed Pegg published the correct number on his website \cite{Pegg05}.}

Now, it is easy to see that the assumption that the events that three points of $\alpha n$ random points in the $n\times n$ grid are in general position are mutually independent is wrong. 
If the points  $\pt{p_1},\pt{p_2},\pt{p_3}$ are collinear and $\pt{p_1},\pt{p_2},\pt{p_4}$ are collinear the event that $\pt{p_1},\pt{p_3},\pt{p_4}$ are collinear is fully determined.
However, even though the assumption is wrong on which they base their calculations to state the conjecture, the result might asymptotically still hold if the dependence of the events is small enough.

In the following, we will briefly review some concepts and tools that have been developed mostly for random graphs in the last 50 years that will help us to properly consider the mutual dependence of certain events concerning our problem.

The key observation is that a random graph on $n$ vertices is nothing else but a random subset chosen according to some distribution of a given set, namely the edge set of the complete graph on $n$ vertices. So the classic binomial and uniform random graph models (see for example \cite{FriezeKaronski16}) and their relation to each other translate directly to respective random (point) set models. We will only state the relevant definitions and facts here. All proofs are completely analogous to the respective proofs for random graphs. For a more detailed account, we refer to \cite{Hainzl20}.

We start by defining uniform random subsets, the type of random subsets Guy and Kelly used in their work, where they choose a random subset of a fixed size uniformly from all subsets of the same size.

\begin{definition}[Random subset models]
Let $\RP_{n,m}$ be the family of all subsets of size~$m$ of a set $S$ with $|S| = n^{2}$. 
To every set $R \in \RP_{n,m}$ we assign the probability
\[\PP(R) = \binom{n^2}{m}^{-1}\]
We denote a random subset chosen according to this distribution by~$R_{n,m}$ and call it a \emph{uniform random subset}.

For the second model, fix $0 \leq p \leq 1$. Then for $0\leq m \leq n^2$, assign to each subset $R$ with $m$ points the probability
\[\PP(R) = p^m(1 - p)^{n^2 - m}\]
We call such a random subset a \emph{binomial random subset} and denote it by $R_{n,p}$.
\end{definition}

The binomial random subset is therefore just a subset, where each element is chosen independently of all others with probability $p$.

Just like the respective random graph models, the random subset models relate to each other such that $R_{n,p}$ conditioned on the event that it contains $m$ elements is distributed like $R_{n,m}$. The proof is completely analogous as in the case of random graphs.

Another helpful concept that we will use are monotone properties. 

\begin{definition}[Monotone properties]
We call a property of a subset $R$ \emph{monotone increasing} if adding an arbitrary element to $R$ does not destroy the property. Conversely, a property is \emph{monotone decreasing} if removing an arbitrary element from $R$ does not destroy the property. 
\end{definition}

By analogous arguments as in~\cite[p.~5ff]{FriezeKaronski16}, one can formally prove what intuitively seems clear: If $P$ is a monotone increasing property, then the probability that a uniform random point set of size~$m'$ is equally or more likely to have the property than a uniform random point set of size $m \leq m'$. The same holds for $R_{n,p}$ and $R_{n,p'}$ if $p \leq p'$. 

\begin{lemma} If $P$ is a monotone increasing property, $p\leq p'$ and $m \leq m'$, then
\[\PP(R_{n,p} \in P) \leq \PP(R_{n,p'} \in P) \quad \text{and} \quad  \PP(R_{n,m} \in P) \leq \PP(R_{n,m'} \in P) \quad \text{respectively} \]
The converse holds for monotone decreasing properties.
\end{lemma}

For monotone properties, the two models satisfy another useful relation, since in a lot of cases it is easier to deal with the binomial model than with the uniform one.

\begin{lemma}\label{lem:equiv}
Let $P$ be a monotone point set property and let $m,n \in \NN$. If we let $p=m/n^2$, then
\[\PP(R_{n,m} \in P) \leq  2 \PP(R_{n,p}  \in P)\]
\end{lemma}

The great advantage of working with binomial random subsets are the available tools that take into account the dependence of certain events. In particular, we will use \emph{Janson's inequality}.

\begin{theorem}[Janson's inequality~\cite{FriezeKaronski16}]\label{thm:Janson}
Let $\mathcal{A}$ be a finite set and $p = \{p_a, a \in \mathcal{A}\}$ be a vector that assigns probabilities to each element in $\mathcal{A}$. Further, let $R_p$ be a random subset in such a way that the elements are chosen independently with $\PP(a\in R_p) = p_a$ for each $a\in A$.

Let $\mathcal{B}$ be a family of subsets of $\mathcal{A}$ and, for every $B \in \mathcal{B}$, let $I_B$ be the indicator variable. Then $X = \sum_{B\in \mathcal{B}} I_B$ counts those elements of $\mathcal{B}$ which are entirely contained in $R_p$. Set
\[ \mu = \EE(X), \quad \Delta = \frac{1}{2} \sum_{\substack{A\neq B \in \mathcal{B} \\ A\cap B \neq \emptyset}} \EE(I_AI_B)\]

\noindent Then, it holds that 
\begin{equation}
    \PP(X = 0) \leq \exp\left(-\mu +\Delta\right)
\end{equation}
If $\Delta > \frac{\mu}{2}$, a stronger bound is given by
\begin{equation}
    \PP(X = 0) \leq \exp\left(-\frac{\mu^2}{\mu + 2\Delta}\right)
\end{equation}
\end{theorem}

In the following, we will use Janson's inequality to bound the probability that a fixed point on the $n\times n$ torus is not dominated. To do so, we will define $\mathcal{A}$ to be the points of the discrete torus and $\mathcal{B}$ to be the set of dominating pairs of this fixed point. A dominating pair of a point $p$ are two points of $T_n$ that define at least one line that dominates $p$. Then, Janson's inequality yields an upper bound on the probability that no dominating pair is contained in the random point set.
Subsequently, we will use the union bound to obtain an upper bound on the probability that there exists a point that is not dominated.

\begin{theorem}[Union bound~\cite{Klenke13}]
Let $\{A_i\}_{i\in\NN}$ be a countable family of events. Then
\[\PP\left(\bigcup_{i \in \NN} A_{i}\right)\leq \sum_{i\in\NN} \PP(A_{i}).\]
\end{theorem}

We are ready to prove the next theorem.

\begin{theorem}\label{thm:torup}
For $n$ prime, $\gd^T_n = O(\sqrt{n \log n})$.
\end{theorem}
\begin{proof}
Let $R_{m}$ be a uniform random point set of size~$m$ on~$T_n$. By Theorem \ref{thm:tordown} and the upper bound on the size of point sets in general position~\cite{Fowler-et-al12}, we can assume $\sqrt{2n} \leq m \leq n+1$.
Now, fix $\pt{p} \in T_n$ and let~$P_\pt{p}$ be the set of all point sets on~$T_n$ that do not dominate~$\pt{p}$. We note that this property is a monotone decreasing point set property.
By Lemma~\ref{lem:equiv}, it holds that
\[\PP(R_{m} \in P_\pt{p}) \leq 2 \PP(R_{n,m/n^2} \in P_\pt{p}).\]
Therefore, we continue by considering $R_{n,m/n^2}$. 

$R_{n,m/n^2}$ dominates~$\pt{p}$ if and only if $R_{n,m/n^2}$ contains a dominating pair of~$\pt{p}$ or $\pt{p}\in R_{n,m/n^2}$. 
So, once again, let~$X$ be the random variable that counts the number of dominating pairs of~$\pt{p}$. That is
\[X=I_\pt{p} + \sum_{D\in \mathcal{D}_\pt{p}} I_D\] 
where~$I_x$ are indicator random variables and $\mathcal{D}_\pt{p}$~denotes the set of dominating pairs of~$\pt{p}$. Note that $X \geq 1$ if and only if $\pt{p}$~is dominated by the random point set. 

Next, we apply Janson's inequality and identify $\mathcal{A}$ with~$T_n$, $R_p$~with~$R_{n,m/n^2}$ and~$\mathcal{B}$ with~$\mathcal{D}_\pt{p}\cup \{\pt{p}\}$. All we have to do is computing~$\mu$ and~$\Delta$. 

Since~$\pt{p}$ is incident to $n+1$~lines that consist of $n$~points, there are exactly $(n+1)\binom{n-1}{2}$~distinct dominating pairs of~$\pt{p}$ on~$T_n$. Consequently,
\begin{align*}
    \mu = \EE(X)&= \PP(\pt{p} \in R_{n,m/n^2}) + \sum_{D\in \mathcal{D}_\pt{p}} \PP(D \in R_{n,m/n^2}) \\
    &= \frac{m}{n^2}+\sum_{D\in \mathcal{D}_\pt{p}} \left(\frac{m}{n^2}\right)^2 =  \frac{m}{n^2} + \frac{(n+1)(n-1)(n-2)}{2}\left(\frac{m}{n^2}\right)^2\\
    &=\frac{m^2}{2n} + \Theta\left(\left(\frac{m}{n}\right)^2\right)
\end{align*}

Furthermore, two distinct dominating pairs~$C,D$ of~$\pt{p}$ have a non-empty intersection if and only if their intersection contains exactly one point and the three points in~$C\cup D$ are collinear. The number of collinear triples that lie on a common line with~$\pt{p}$ is~$(n+1)\binom{n-1}{3}$. Hence,
\begin{align*}
    \Delta &= \frac{1}{2} \sum_{\substack{C\neq D \in \mathcal{D}_\pt{p}\cup \{\pt{p}\} \\ C\cap D \neq \emptyset}} \EE(I_CI_D) = \frac{1}{2}  \sum_{\substack{\pt{x},\pt{y},\pt{z} \in [n]^2 \\\pt{x},\pt{y},\pt{z}, \pt{p} \text{ collinear}}} \left(\frac{m}{n^2}\right)^3 \\
    &= \frac{(n+1)(n-1)(n-2)(n-3)}{12}\left(\frac{m}{n^2}\right)^3\\
    &= \frac{m^3}{12n^2} + \Theta\left( \left( \frac{m}{n}\right)^3\right)
\end{align*}

Using Janson's inequality, we derive
\begin{align*}
    \PP(R_{m} \in P_\pt{p}) &\leq 2 \,\PP(R_{n,m /n^2} \in P_\pt{p}) = 2\, \PP(X = 0) \leq 2\, \exp(-\mu + \Delta)\\
    &= 2\, \exp\left(-\frac{m^2}{2n} + \Theta\left(\left(\frac{m}{n}\right)^2\right) + \frac{m^3}{12n^2} + \Theta\left( \left( \frac{m}{n}\right)^3\right)\right)
\end{align*}

The sum of the asymptotic terms above is~$O(1)$ because we assumed that $m\leq n+1$. Finally, we use the union bound to show that the probability that there exists a free point (that is a point which is not dominated by $R_m$) is smaller than 1. Consequently, there has to be a realization of $R_m$ where all points are dominated.
\begin{align*}
    \PP(\exists \text{ a free point in } T_n) &= \PP(R_{m} \in \bigcup_{\pt{p} \in [n]^2} P_\pt{p}) \\
    &\leq \sum_{\pt{p}\in [n]^2} \PP(R_{m} \in P_\pt{p}) \\
    &\leq 2 n^2 \exp\left(-\frac{m^2}{2n} + \frac{m^3}{12n^2} + O(1)\right)\\
    &=\exp\left(2 \log n - \frac{m^2}{2n} + \frac{m^3}{12n^2} + O(1)\right)
\end{align*}
If we choose $m = (2+\varepsilon)\sqrt{n\log n}$, with $\varepsilon>0$, it follows that
\[2 \log n - \frac{m^2}{2n} + \frac{m^3}{12n^2} + O(1) \rightarrow -\infty \quad \text{as } n \rightarrow \infty\]
Consequently, for all $n$ large enough, there is a dominating set of size~$O(\sqrt{n \log n})$.
\end{proof}

\begin{figure}[h]
    \centering
    \includegraphics[page=4, width=0.24\textwidth]{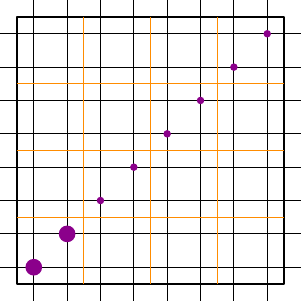}
    \includegraphics[page=5, width=0.24\textwidth]{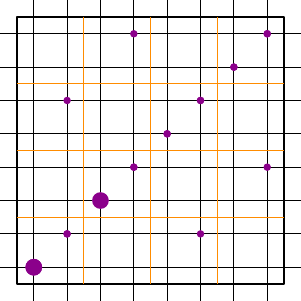}
    \includegraphics[page=6, width=0.24\textwidth]{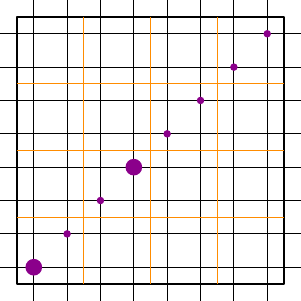}
    \includegraphics[page=7, width=0.24\textwidth]{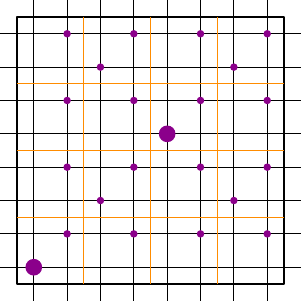}
    \caption{Blowing up a point set in $T_{2}$ to a corresponding one in $T_{8}$ by the factors $2,3$ and $4$. Notice how factors like $2$ or $3$ would yield a set dominating less points than the one obtained by the blow up factor $4$.}
    \label{fig:meshes}
\end{figure}

We can prove further bounds for arbitrary composite numbers $n$. 
The main idea is to take always a factor $p$ of $n$ and divide $T_n$ into $p^2$ grids of size $n/p\times n/p$, fix a point $x$ in a dominating set $S$ of $T_p$ (which we will assume to be $(0,0)$) and blow it up by the factor $n/p$. Note that this is equivalent to refining a $p \times p$ grid on $T_p$ by a factor of $n/p$ while keeping the points of $S$ at the same position (See Figure \ref{fig:meshes} for the case $n=2^3$ and blow up factor $4$).
By elementary number theory, it follows that the lines that are incident to $x$ and another point in the blown up copy dominate (nearly) all points that are congruent to points in $T_p$ that are dominated by the line incident to $x$ and the corresponding point in $S$. Or in other words, if two points are congruent to each other modulo~$p$, the lines that they span with~$(0,0)$ will bundle in all points that are congruent to~$(0,0)$.

\begin{lemma}\label{L:blowup}
    Let $n = p\cdot q$, where $p$ is prime and $(y_1,y_2) \in T_p\setminus \{(0,0)\}$.\\
    If $(0,0)$ and $(y_1,y_2)$ dominate the point $(a_1,a_2)\neq (0,0)$ on the $p\times p$ torus $T_p$
    then $(0,0)$ and $\pi_{n}(y_1q,\;y_2 q)$ dominate all points $(b_1,b_2)$ on $T_n$ with $\pi_p(b_1, b_2) = (a_1,a_2)$.
\end{lemma}

\begin{proof}
    Without loss of generality, let $a_1 \neq 0$.
    Let us recall that, if $p$ is prime, there is a unique $1\leq A < p$ such that $A a_1 = 1 \Mod p$. We also know that there is a unique line defined by $(0,0)$ and $(y_1,y_2)$ if $p$ is prime, which will be of the form
    \[L = \left\{(a,b) \in T_p \mid  b=aAa_2 \Mod p\right\}\]
    
    Since it dominates $(y_1,y_2)$, it holds that 
    \begin{equation}\label{eq:y_2}
        y_2 = y_1 A a_2 \Mod p
    \end{equation}
    
    Further, as $\pi_p(b_1,b_2) = (a_1,a_2)$ there are $t_1, t_2 \in \ZZ$ such that
    \[b_1 = a_1 + t_1p \quad b_2 = a_2 + t_2p\]
    
    Now, we consider the set
    \[L_n = \left\{(a,b) \in T_n \mid  (a,b) = \pi_n(t(b_1,b_2)), t \in \ZZ\right\}\]

    that is clearly a subset of at least one line incident to $(0,0)$ and $(b_1,b_2)$ on~$T_n$. We want to show, that it is also incident to $q(y_1,y_2)$, which is the case if there is a $t$ such that $t(b_1,b_2) = q(y_1,y_2) \Mod n$. Notice that
    \[q(y_1,y_2) = q(y_1,y_1Aa_2 + t_{y}p) = q(y_1,y_1Aa_2) \Mod n\]
    So the choice $t=qy_1A$ yields
    \begin{align*}
    qy_1A(b_1, b_2) &= qy_1A(a_1 + t_1p, a_2 + t_2p) \Mod n\\
    &= qy_1(1 + t_1'p, Aa_2 + t_2'p) \Mod n\\
    &= q(y_1, y_1Aa_2) \Mod n\\
    &= q(y_1,y_2) \Mod n
    \end{align*}
    which was exactly what we wanted to show.
\end{proof}

Now, one thing we should note about Lemma \ref{L:blowup} is that $(0,0)$ is quite an arbitrary choice as a fixed point when we blow up the set. The effect depends entirely on the distance between $(0,0)$ and $q(y_{1},y_{2})$. To see this, remember that a line on the torus $T_{n}$ is just a projection of a line in $\ZZ^{2}$ modulo $n$. So, we can think about $\ZZ^{2}$ as a tiling of copies of the torus ($\ZZ^{2}$ is the universal cover of the discrete torus). Now we could just choose a different cutout (fundamental domain) for $T_{n}$ and project onto this cutout modulo $n$. See Figur \ref{fig:cutout} for an illustration. Technically, one could also repeat the exact same computations in the proof of Lemma \ref{L:blowup} to prove the following corollary. 
\begin{figure}
\centering
\includegraphics[width=0.8\textwidth, page=2]{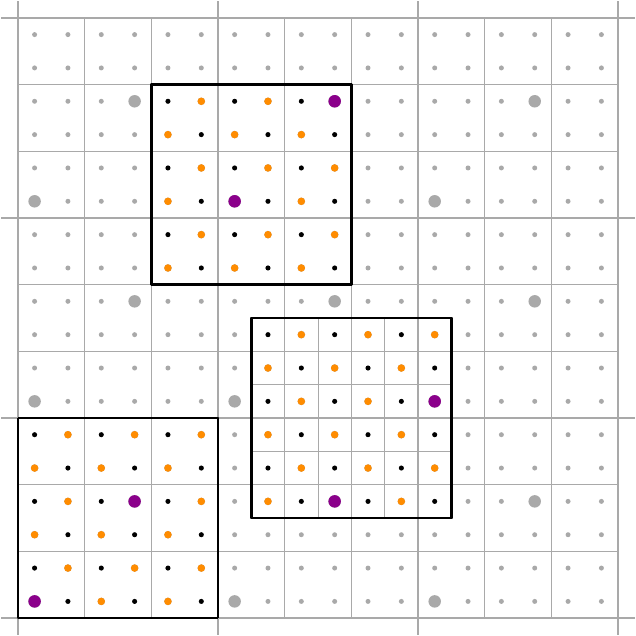}
\caption{Choosing different cutouts for $T_{6}$. Orange marked points are always dominated by the magenta points on $T_{6}$}\label{fig:cutout}
\end{figure}

\begin{corollary}\label{C:blowup}
    Let $n = p\cdot q$, where $p$ is prime and $(x_{1},x_{2}) \neq (y_1,y_2) \in T_p$.\\
    If $(x_{1},x_{2})$ and $(y_1,y_2)$ dominate the point $(a_1,a_2)\neq (x_{1},x_{2})$ on the $p\times p$ torus $T_p$, then $(x_{1},x_{2})$ and $\pi_{n}(x_1+q(y_{1}-x_{1}),\;x_2 +q(y_{1}-x_{1}))$ dominate all points $(b_1,b_2)$ on $T_n$ with $\pi_p(b_1, b_2) = (a_1,a_2)$.
\end{corollary}

Another thing worth noting is that given $(0,0)$ and $q(y_{1},y_{2})$, we can see them as a blown up copy of $(0,0)$ and $(y_{1},y_{2})$ or as a blown up copy of $q(y_{1},y_{2})$ and $(q-1)(y_{1},y_{2})$. If $ q(y_{1}, y_{2})) \neq (0,0) \Mod p$, we can apply Corollary \ref{C:blowup} to the second pair of points with $q(y_{1},y_{2})$ as a fixed point and deduce that all points congruent to $(0,0)$ are dominated.\\

We will use these observations in the proofs of the following theorems. Nevertheless it might be even more useful for the design of algorithms that look for dominating sets on $T_n$. 

We start by proving a general upper bound if $n$ is a composite number.

\begin{theorem}
    Let $n=pq$, where $p$ is prime. If $\gcd (p,q) = 1$, then
    \[\gd_n^T \leq 2p\]
    otherwise
    \[\gd_{n}^{T} \leq 2(p+1)\]
\end{theorem}
\begin{proof}
Choose a trivial dominating set of $T_p$ with two points in each column. That is, for example, 
\[S = \{(a,b) \in T_p \mid a \in \{0,1\}, 0 \leq b \leq p-1\}\]
Now, blow up only pairs of this set, meaning that a pair $\{(0,a),(1,a)\}$ corresponds to the pair $\{(0,a), (q,a)\}$ which dominates all points $(b_1,b_2)$ with $\pi_p(b_1,b_2) = (b,a)$, $1\leq b \leq p-1$ by Corollary \ref{C:blowup}. If $\gcd(p,q) = 1$, then $\pi_{p}(q,a) \neq (0,a)$ and consequently, we can apply Corollary \ref{C:blowup} with $(q,a)$ as a fixed point and deduce all points congruent to $(0,a)$ are dominated as well.\\
If $\gcd(p,q) > 1$, we need to add further points to dominate all of these remaining points. Thus, we add another point at $(0,q)$. In combination with $(0,0)$, it  dominates all points $(b_1,b_2)$ with $\pi_p(b_1,b_2) = (0,a)$, $1\leq a\leq p-1$ and another at $(0,1+q)$ which, in combination with $(0,1)$ dominates all points congruent to $(0,0)$ modulo $p$.
\end{proof}

Now, given the fact that we know that there exists a dominating set of size $c\sqrt{p \log p}$ of $T_p$ if $p$ is prime, it is reasonable to assume that one can find a dominating set of size $o(n)$ that satisfies the condition of the following theorem and would give a better upper bound than the previous theorem.

\begin{theorem}\label{thm:exact}
    Let $n = pq$, where $p$ is prime and let $S$ be a dominating set of $T_p$ and $\pt{x}\in S$ such that every point in $T_p$ is dominated by a line incident to $\pt{x}$ and another point in~$S$. Then 
    \[\gd_n^T \leq |S|+2.\]
    If $n=p^k$, then it holds
    \[\gd_n^T \leq |S|.\]
\end{theorem}

\begin{proof}
    Without loss of generality, let $\pt{x}=(0,0)$ and we define the blown up dominating set of $S$ to be $S' = \{qy \mid y \in S\}$. If every point is dominated by a line incident to $\pt{x}\in S$, we know by Lemma \ref{L:blowup} that the blown up copy $S'$ will dominate all points in $T_n$ that are not congruent to $(0,0)$.
    
    Now, we place another two points, one at $(0,1)$ and one at $(0,q+1)$. By Corollary \ref{C:blowup}, these two points will dominate all points which are congruent to $(0,b), b \neq 1$ including those that are congruent to $(0,0)$ modulo $p$.
    
    If $n=p^k$, we do not need to place additional points, since the subgrid of all points that are congruent to $(0,0)$ modulo $p$ contains the set $S'$ as a blown up copy of $S$ (by a factor of $p^{k-2}$. Note that this is equivalent to coarsening the grid on $T_n$ by a factor of $p$. See Figure \ref{fig:meshes}). By Lemma \ref{L:blowup} we can conclude that all points that are not congruent to $(0,0)$ modulo $p^2$ are dominated and continue this argument until we end up with a $p\times p$ grid that is dominated by $S$ by assumption.
\end{proof}

\begin{figure}[h]
    \centering
    \includegraphics[page=3, width=0.3\textwidth]{8x8torus.pdf}
    \includegraphics[page=1, width=0.3\textwidth]{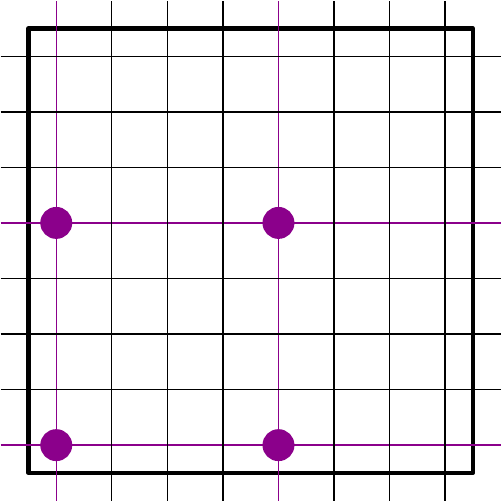}
    \includegraphics[page=2,width=0.3\textwidth]{8x8torus.pdf}
    \caption{Refining the $2\times2$ discrete torus to an $8\times8$ torus and back to a $4\times4$ torus}
    \label{fig:meshes}
\end{figure}

As an example, we can improve the upper bound for even numbers.

\begin{corollary}
If $n$ is even, it holds that $\gd_{n}^T \leq 4$.
\end{corollary}
\begin{proof}
    We choose $S= T_2$. If $n=2^k$, this is a direct consequence of Theorem \ref{thm:exact}.

    For $n=2^kq$, where $q$ is odd, we can coarsen the grid on the torus by a factor of $2^{k-1}$ by the same arguments as in the proof of Theorem \ref{thm:exact} and only consider the torus $T_n$, where $n=2q$ with $q$ odd. Again, we know that all points that are not congruent to $(0,0)$ are dominated. \\
    Since $q$ is odd, $\pi_{p}(q,q) = (1,1)$ and we can apply Corollary \ref{L:blowup} to the situation where $(q,q)$ is the fixed point. Therefore we know that all points that are congruent to $(0,0)$ are dominated as well.
\end{proof}

So far, we have hardly exploited the observation that we can use Lemma \ref{L:blowup} in two directions if $\gcd(p,q)=1$. Here is another result one can derive from it, but once again, it could come in particularly useful in the design of algorithms that search for dominating sets. 

\begin{proposition}
If $n = 3q$, it holds that $\gd_{n}^T \leq 4$.
\end{proposition}
\begin{proof}
    We choose $S=\{(0,0),(0,t),(t,0),(t,t)\}$ with $t= q \Mod 3$ and we assume first that $\gcd(3,q)=1$. That is, $t$ is either $1$ or $2$ depending on the congruence class of $q$ and in both cases, if we blow up $S$ by $q$, we end up with a set that contains points that are congruent to $\{(0,0),(0,1),(1,0),(1,1)\}$ modulo $3$. 

    We know by similar argumentation as in previous proofs that all points which are congruent to $(0,0),(0,1),(0,2), (1,0),(2,0),(1,1)$ or $(2,2)$ are dominated and we are left to show that points congruent to $(1,2)$ and $(2,1)$ are dominated as well. However, given that we have two points that are essentially a blown up copy of $(0,1)$ and $(1,1)$, we know that all points congruent to $(2,1)$ are dominated, as we can apply Corollary \ref{C:blowup} once again and choose for example $(0,tq)$ as the fixed point. The same works by symmetry for points congruent to $(1,2)$.\\
    
    In the case where $\gcd(3,q)>1$, the points in the blown up copy satisfy $(0,0) = (0,qt) = (qt,0) = (qt,qt) \Mod 3$. However, $(0,qt)$ and $(qt,0)$ are a blown up copy of points congruent to $(0,0)$ and $(1,-1)=(1,2) \Mod 3$, which dominate all points congruent to $(1,2)$ and $(2,1)$. Therefore, we can coarsen the grid again until we are in the situation where $\gcd(3,q)=1$.
\end{proof}

\section{Variations of the Problem and Conclusion}

Another interesting variant of the geometric dominating set problem is a game version: Two players alternatingly place a point on the $n \times n$-grid such that no three points are collinear. The last player who can place a valid point wins the game (this is called normal play in game theory). It is not hard to see that for any even $n$ the second player has a winning strategy. She just always sets the point which is center mirrored to the previous move of the first player. By symmetry arguments this move is always valid, as long as the first player made a valid move. For $n$ odd the situation is more involved. If the first player does not start by placing the central point in her first move, then we have again a winning strategy for the second player by the same reasoning (note that the central place can not be used after the first two points have been placed, as this would cause collinearity). So if the first player starts by placing the central point it can be shown that for $n=3$ she can also win the game. But for $n=5,7,9$ still the second player has a winning strategy. For odd $n$ we currently do not know the outcome for games on grids of size $n\geq11$.

Several open problems arise from our considerations:
\begin{itemize}
\item Is there a constant $c>0$ such that $\gd_n \leq \gdd_n \leq (2-c)n$ holds for large enough $n$?
\item Do $\gdd_n$ and $\gd_n$ grow in a monotone way, that is, is  $\gdd_{n+1} \geq \gdd_n$ and $\gd_{n+1} \geq \gd_n$?
\item Is there some $n_0$ such that $\gdd_n > \gd_n$ for all $n \geq n_0$? 
\item Do minimal dominating sets in general position always have even cardinality? For $n \leq 12$ this is the case, but the currently best example for $n=17$ might be a counterexample. 
\item How much can the size of dominating sets (with or without collinear points) be improved if the points are allowed to lie outside the grid?
\item Which player has a winning strategy in the game version for boards of size $n\geq11$, $n$ odd?
\item For which values of $n$ can we compute the exact domination number of the $n\times n$ torus?
\end{itemize}

\subparagraph*{Acknowledgments.} 
This research was initiated at the $33^{rd}$ Bellairs Winter Workshop on Computational Geometry in 2018, and continued at the 2019 edition of this workshop. We thank the organizers and participants of both workshops for a very fruitful atmosphere.

\bibliography{DomSets_Bib}{}

\end{document}